\documentclass[11pt]{article}
\usepackage{times}
\usepackage{latexsym}
\usepackage{epsfig,enumerate,amsmath,amsfonts,amssymb,setspace}
\usepackage{indentfirst}
\usepackage{wrapfig,lipsum}
\usepackage{epsfig}
\usepackage{graphicx}
\usepackage{amsthm}
\usepackage{algorithm}
\usepackage[T1]{fontenc}
\usepackage{authblk}
\usepackage[algo2e,ruled,vlined]{algorithm2e}
\newtheorem{definition}{Definition}[section]
\newtheorem{theorem}{Theorem}[section]
\newtheorem{lemma}{Lemma}[section]

\newtheorem{obs}{Observation}[section]
\newtheorem{claim}{Claim}[section]

\title{Complexity and Algorithms for Semipaired Domination in Graphs}
\author[1]{Michael A. Henning \thanks{mahenning@uj.ac.za}}
\author[2]{Arti Pandey\thanks{arti@iitrpr.ac.in}}
\author[2]{Vikash Tripathi\thanks{2017maz0005@iitrpr.ac.in}}
\affil[1]{Department of Pure and Applied Mathematics, University of Johannesburg, Auckland Park, 2006 South Africa}
\affil[2]{Department of Mathematics, Indian Institute of Technology Ropar, Nangal Road, Rupnagar, Punjab 140001, INDIA}

\begin{document}
\date{}
\maketitle
\begin{abstract}
For a graph $G=(V,E)$ with no isolated vertices, a set $D\subseteq V$ is called a semipaired dominating set of G if  $(i)$ $D$ is a dominating set of $G$, and $(ii)$ $D$ can be partitioned into two element subsets such that the vertices in each two element set are at distance at most two. The minimum cardinality of a semipaired dominating set of $G$  is called the semipaired domination number of $G$, and is denoted by $\gamma_{pr2}(G)$. The \textsc{Minimum Semipaired Domination} problem is to find a semipaired dominating set of $G$ of cardinality $\gamma_{pr2}(G)$. In this paper, we initiate the algorithmic study of the \textsc{Minimum Semipaired Domination} problem. We show that the decision version of the \textsc{Minimum Semipaired Domination} problem is NP-complete for bipartite graphs and split graphs. On the positive side, we present a linear-time algorithm to compute a minimum cardinality semipaired dominating set of interval graphs and trees. We also propose a $1+\ln(2\Delta+2)$-approximation algorithm for the \textsc{Minimum Semipaired Domination} problem, where $\Delta$ denote the maximum degree of the graph and show that the \textsc{Minimum Semipaired Domination} problem cannot be approximated within $(1-\epsilon) \ln|V|$ for any
$\epsilon > 0$ unless NP $\subseteq$ DTIME$(|V|^{O(\log\log|V|)})$.
\end{abstract}

\vspace*{.2cm}

\section{Introduction}
A \emph{dominating set} in a graph $G$ is a set $D$ of vertices of $G$ such that every vertex in $V(G) \setminus D$ is adjacent to at least one vertex in $D$. The \emph{domination number} of $G$, denoted by $\gamma(G)$, is the minimum cardinality of a dominating set of $G$. The \textsc{Minimum Domination} problem is to find a dominating set of cardinality $\gamma(G)$. More thorough treatment of domination, can be found in the books~\cite{hhs1,hhs2}.
A dominating set $D$ is called a \emph{paired dominating set} if $G[D]$ contains a perfect matching. The \emph{paired domination number} of $G$, denoted by $\gamma_{pr}(G)$ is the minimum cardinality of paired dominating set of $G$. The concept of paired domination was introduced by Haynes and Slater in \cite{paired}.

A relaxed form of paired domination called semipaired domination was introduced by Haynes and Henning~\cite{semi-paired1} and studied further in \cite{semi-paired2, semi-paired3, semi-paired4}. A set $S$ of vertices in a graph $G$ with no isolated vertices is a \emph{semipaired dominating set}, abbreviated a semi-PD-set, of $G$ if $S$ is a dominating set of $G$ and $S$ can be partitioned into $2$-element subsets such that the vertices in each $2$-element set are at distance at most $2$.  In other words, the vertices in the dominating set $S$ can be partitioned into $2$-element subsets such that
if $\{u, v\}$ is a $2$-set, then the distance between $u$ and $v$ is either $1$ or $2$. We say
that $u$ and $v$ are \emph{semipaired}. The \emph{semipaired domination number} of $G$, denoted by $\gamma_{pr2}(G)$, is the minimum cardinality of a semi-PD-set of $G$. Since every paired dominating set is a semi-PD-set, and every semi-PD-set is a dominating set, we have the following observation.

\begin{obs}{\rm (\cite{semi-paired1})}
 \label{ob:chain}
For every isolate-free graph $G$, $\gamma(G) \le \gamma_{pr2}(G) \le \gamma_{pr}(G)$.
\end{obs}

By Observation~\ref{ob:chain}, the semipaired domination number is squeezed between two fundamental domination parameters, namely the domination number and the paired domination number.

More formally, the minimum semipaired domination problem and its decision version are defined as follows:

\noindent\underline{\textsc{Minimum Semipaired Domination} problem (MSPDP)}
\\
[-12pt]
\begin{enumerate}
  \item[] \textbf{Instance}: A graph $G=(V,E)$.
  \item[] \textbf{Solution}: A semi-PD-set $D$ of $G$.
  \item[] \textbf{Measure}: Cardinality of the set $D$.
\end{enumerate}

\noindent\underline{\textsc{Semipaired Domination Decision} problem (SPDDP)}
\\
[-12pt]
\begin{enumerate}
  \item[] \textbf{Instance}: A graph $G=(V,E)$ and a positive integer $k \le |V|$.
  \item[] \textbf{Question}: Does there exist a semi-PD-set $D$ in $G$ such that $|D| \le k$?
\end{enumerate}

In this paper, we initiate the algorithmic study of the semipaired domination problem. The main contributions of the paper are summarized below. In Section~\ref{Sec:2}, we discuss some definitions and notations.
In Section~\ref{Sec:3}, we discuss the difference between the complexity of paired domiantion and semipaired domination in graphs. In Section~\ref{Sec:4}, we show that the \textsc{Semipaired Domination Decision} problem is NP-complete for bipartite and split graphs. In Section~\ref{Sec:5} and Section~\ref{Sec:6}, we propose a linear-time algorithms to solve the \textsc{Minimum Semipaired Domination} problem in interval graphs and trees respectively. In Section~\ref{Sec:7}, we propose an approximation algorithm for the \textsc{Minimum Semipaired Domination} problem in general graphs. In Section~\ref{Sec:8}, we discuss an approximation hardness result. Finally, Section~\ref{Sec:9}, concludes the paper.

\section{Terminology and Notation}
\label{Sec:2}

For notation and graph theory terminology, we in general follow~\cite{total2}. Specifically, let $G = (V, E)$ be a graph with vertex set $V=V(G)$ and edge set $E=E(G)$, and let $v$ be a vertex in $V$. The \emph{open neighborhood} of $v$ is the set $N_G(v) = \{u \in V \, | \, uv \in E\}$ and the \emph{closed neighborhood of $v$} is $N_G[v] = \{v\} \cup N_G(v)$. Thus, a set $D$ of vertices in $G$ is a dominating set of $G$ if $N_G(v) \cap D \ne \emptyset$ for every vertex $v \in V \setminus D$, while $D$ is a total dominating set of $G$ if $N_G(v) \cap D \ne \emptyset$ for every vertex $v \in V$. The \emph{distance} between two vertices $u$ and $v$ in a connected graph $G$, denoted by $d_G(u,v)$, is the length of a shortest $(u,v)$-path in $G$.
If the graph $G$ is clear from the context, we omit it in the above expressions.
We write $N(v)$, $N[v]$ and $d(u,v)$ rather than $N_G(v)$, $N_G[v]$ and $d_G(u,v)$, respectively.

For a set $S \subseteq V(G)$, the subgraph induced by $S$ is denoted by $G[S]$.  If $G[C]$, where $C \subseteq V$, is a complete subgraph of $G$, then $C$ is a \emph{clique} of $G$. A set $S \subseteq V$ is an \emph{independent set} if $G[S]$ has no edge. A graph $G$ is \emph{chordal} if every cycle in $G$ of length at least four has a \emph{chord}, that is, an edge joining two non-consecutive vertices of the cycle.
A chordal graph $G=(V,E)$ is a \emph{split graph} if $V$ can be partitioned into two sets $I$ and $C$ such that $C$ is a clique and $I$ is an independent set.
A vertex $v\in V(G)$ is a \emph{simplicial} vertex of $G$ if $N_G[v]$
is a clique of $G$. An ordering $\alpha=(v_1,v_2,...,v_n)$ is a {\it
perfect elimination ordering} (PEO) of vertices of $G$ if $v_i$ is a simplicial
vertex of $G_i=G[\{v_i,v_{i+1},...,v_n\}]$ for all $i$, $1\leq i\leq
n$. Fulkerson and Gross~\cite{gross} characterized chordal graphs, and showed that a graph $G$ is chordal if and only if
it has a PEO. A graph $G = (V,E)$ is \emph{bipartite} if $V$ can be partitioned into two disjoint sets $X$  and $Y$ such that every edge of $G$ joins a vertex in $X$ to a vertex  in $Y$, and such a partition $(X,Y)$ of $V(G)$ is called a \emph{bipartition} of $G$. Further, we denote such a bipartite graph $G$ by $G=(X,Y,E)$. A graph $G$ is an \emph{interval graph} if there exists a one-to-one correspondence between its vertex set and a family of closed intervals in the real line, such that two vertices are adjacent if and only if their corresponding intervals intersect. Such a family of intervals is called an \emph{interval model} of a graph.

In the rest of the paper, all graphs considered are simple connected graphs with at least two vertices, unless otherwise mentioned specifically. We use the standard notation $[k] = \{1,\ldots,k\}$.
For most of the approximation related terminologies, we refer to  \cite{ausiello,klasing}.

\section{Complexity difference between paired domination and semipaired domination}
\label{Sec:3}
In this section, we make an observation on complexity difference between paired domination and semipaired domination. We show that the decision version of the \textsc{Minimum paired domination} problem is NP-complete for GP$4$ graphs, but the \textsc{Minimum Semipaired Domination} problem is easily solvable for GP$4$ graphs. The class of GP$4$ graphs was introduced by Henning and Pandey in \cite{semito}. Below we recall the definition of GP$4$ graphs.

\begin{definition}[GP$4$-graph]
\label{defn1}
A graph $G=(V,E)$ is called a \emph{GP}$4$-\emph{graph} if it can be obtained from a general connected graph $H=(V_{H},E_{H})$ where $V_{H}=\{v_{1},v_{2},\ldots,v_{n_{H}}\}$, by adding a path of length~$3$ to every vertex of $H$. Formally, $V = V_{H} \cup \{ w_{i},x_{i},y_{i},z_{i} \mid 1\leq i \leq n_{H} \, \}$ and $E=E_{H}\cup \{v_{i}w_{i},w_{i}x_{i},x_{i}y_{i},y_{i}z_{i}\mid 1\leq i \leq n_H \, \}$.
\end{definition}

\begin{theorem}
 \label{ob:GP4}
If $G$ is a \emph{GP}$4$-\emph{graph}, then $\gamma_{pr2}(G) = \frac{2}{5}|V(G)|$.
\end{theorem}

\begin{lemma}
\label{l:lemGP4}
If $G$ is a GP$4$-graph constructed from a graph $H$ as in Definition~\ref{defn1}, then $H$ has a paired dominating set of cardinality~$k$, $k\leq n_H$ if and only if $G$ has a semi-PD-set of cardinality $2 n_H +k$.
\end{lemma}

Since the decision version of the \textsc{Minimum Paired Domination} problem is known to be NP-complete for general graphs~\cite{paired}, the following theorem follows directly from Lemma~\ref{l:lemGP4}.

\begin{theorem}
The decision version of the \textsc{Minimum Paired Domination} problem is NP-complete for GP$4$-graphs.
\end{theorem}

\newpage
\section{NP-completeness Results}
\label{Sec:4}
In this section, we study the NP-completeness of the \textsc{Semipaired Domination Decision} problem. We show that the \textsc{Semipaired Domination Decision} problem is NP-complete for bipartite graphs and split graphs.

\subsection{NP-completeness proof for bipartite graphs}
\begin{theorem}\label{t:bipartite}
The \textsc{Semipaired Domination Decision} problem is NP-complete for bipartite graphs.
\end{theorem}
\begin{proof}
Clearly, the \textsc{Semipaired Domination Decision} problem is in NP for bipartite graphs. To show the hardness, we give a polynomial reduction from the \textsc{Minimum Vertex Cover} problem. Given a non-trivial graph $G=(V,E)$, where $V=\{v_{1},v_{2},\ldots,v_{n}\}$ and $E=\{e_{1},e_{2},\ldots,e_{m}\}$, we construct a graph $H=(V_{H},E_{H})$ in the following way:

 Let $V_{k}=\{v_{i}^{k} \mid i \in [n]\}$ and $E_{k}=\{e_{j}^{k}\mid j \in [m]\}$ for $k\in [2]$. Also assume that $A=\{a_{i}\mid i \in [n]\}$, $B=\{b_{i}\mid i \in [n]\}$, $C=\{c_{i} \mid i \in [n]\}$, and $F=\{f_{i}\mid i \in [n]\}$.

 Now define $V_{H}=V_{1}\cup V_{2}\cup E_{1}\cup E_{2} \cup A\cup B \cup C \cup F$,\\
 and $E_{H}=\{v_{i}^{1}f_{i}, v_{i}^{2}f_{i},a_{i}b_{i},b_{i}c_{i},a_{i}f_{i}\mid i \in [n]\} \cup \{v_{p}^{k}e_{i}^{k},v_{q}^{k}e_{i}^{k}\mid k\in [2], i \in [m]$, $v_{p},v_{q}$ are endpoints of edge $e_{i}$ in $G\}$. Fig.~\ref{fig:1} illustrates the construction of $H$ from $G$.

 \begin{figure}[h!]
 \begin{center}
  \includegraphics[width=12cm, height=5cm]{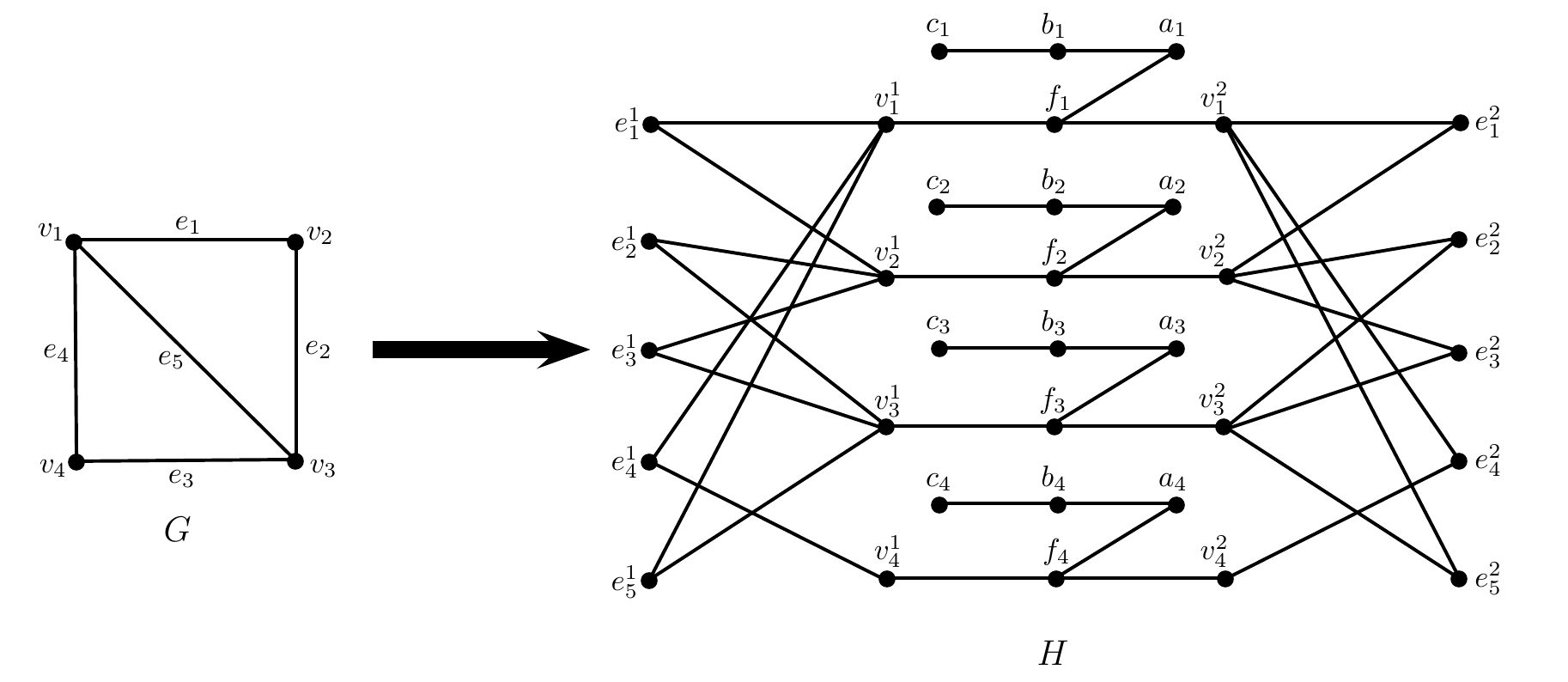}\\
 \caption{An illustration of the construction of $H$ from $G$  in the proof of Theorem~\ref{t:bipartite}.}
\label{fig:1}
\end{center}
\end{figure}

 Note that the set $I_{1}=V_{1}\cup V_{2}\cup A \cup C$ is an independent set in $H$. Also, the set $I_{2}=E_{1}\cup E_{2} \cup F \cup B$ is an independent set in $H$. Since $V_{H}=I_{1}\cup I_{2}$, the graph $H$ is a bipartite graph. Now to complete the proof, it suffices for us to prove the following claim:

\begin{claim}
\label{c:claim1}
The graph $G$ has a vertex cover of cardinality at most~$k$ if and only if the graph $H$ has a semi-PD-set of cardinality at most~$2n+2k$.
\end{claim}
\begin{proof}
Let $V_{c}=\{v_{i_1},v_{i_2},\ldots,v_{i_k}\}$ be a vertex cover of $G$ of cardinality $k$. Then $D_{p}=\{v_{i_1}^{1},v_{i_2}^{1},\ldots,v_{i_k}^{1}\}\cup \{v_{i_1}^{2},v_{i_2}^{2},\ldots,v_{i_k}^{2}\}\cup B\cup F$ is a semi-PD-set of $H$ of cardinality $2n+2k$.

Conversely, suppose that $H$ has a semi-PD-set $D$ of cardinality at most $2n+2k$. Note that $D\cap \{a_{i}.b_{i},c_{i},f_{i}\}|\geq 2$ for each $i\in [n]$. Hence, without loss of generality, we may assume that $\{b_{i},f_{i}\mid i \in [n]\} \subseteq D$, where $b_{i}$ and $f_{i}$ are semipaired. Hence $|D\cap (E_{1}\cup E_{2} \cup V_{1} \cup V_{2})|\leq 2k$. Let $S=(V_{1}\cup E_{1})\cap D$.  Without loss of generality, we may also assume that $|S|\leq k$. Now, if $e_{i}^{1}\in S$ for some $i\in [m]$, and none of its neighbors belongs to $D$, then $e_{i}^{1}$ must be semipaired with some vertex $e_{j}^{1}$ where $j\in [m]\setminus \{i\}$, and also there must exists a vertex $v_{k}^{1}$ which is a common neighbor of $e_{i}^{1}$ and $e_{j}^{1}$. In this case, we replace
the vertex $e_{i}^{1}$ in the set $S$ with the vertex $v_{k}^{1}$ and so
$S \leftarrow (S\setminus \{e_{i}^{1}\}) \cup \{v_{k}^{1}\}$ where $v_{k}^{1}$ and $e_{j}^{1}$ are semipaired. We do this for each vertex $e_{i}^{1} \in S$ where $i \in [m]$ with none of its neighbors in the set $D$. For the resulting set $S$, $|S \cap V_{1}| \leq k$ and every vertex $e_{i}^{1}$ has a neighbor in $V_{1}\cap S$. The set $V_{c}=\{v_{i} \mid v_{i}^{1}\in S\}$ is a vertex cover of $G$ of cardinality at most $k$. This completes the proof of the claim.
\end{proof}
Hence, the theorem is proved.
\end{proof}
\subsection{NP-completeness result for split graphs}
\begin{theorem}\label{t:split}
The \textsc{Semipaired Domination Decision} problem is NP-complete for split graphs.
\end{theorem}

\begin{proof}
Clearly, the \textsc{Semipaired Domination Decision} problem is in NP. To show the hardness, we give a polynomial time reduction from the \textsc{Domination Decision} problem, which is well known NP-complete problem. Given a non-trivial graph $G=(V,E)$, where $V=\{v_{i}\mid i\in [n]\}$ and $E=\{e_{j}\mid j\in [m]\}$, we construct a split graph $G'=(V_{G'},E_{G'})$ as follows:

Let $V_{k}=\{v_{i}^{k} \mid i \in [n]\}$ and $U_{k}=\{u_{i}^{k}\mid i \in [n]\}$ for $k\in [2]$.
Now define $V_{G'}=V_{1}\cup V_{2}\cup U_{1}\cup U_{2}$, and $E_{G'}=\{uv \mid u,v \in V_1 \cup U_1, u\neq v\} \cup \{v_i^2v_j^1,u_i^2u_j^1 \mid i \in [n] $ and $v_j \in N_G[v_i]\}$. Note that the set $A=V_1 \cup U_1$ is a clique in $G'$ and the set $B=V_2 \cup U_2$ is an independent set in $G'$. Since $V_{G'}=A \cup B$, the constructed graph $G'$ is a split graph. Fig.~\ref{fig:1} illustrates the construction of $G'$ from $G$.

 \begin{figure}[h!]
 \begin{center}
 \includegraphics[width=8cm, height=4cm]{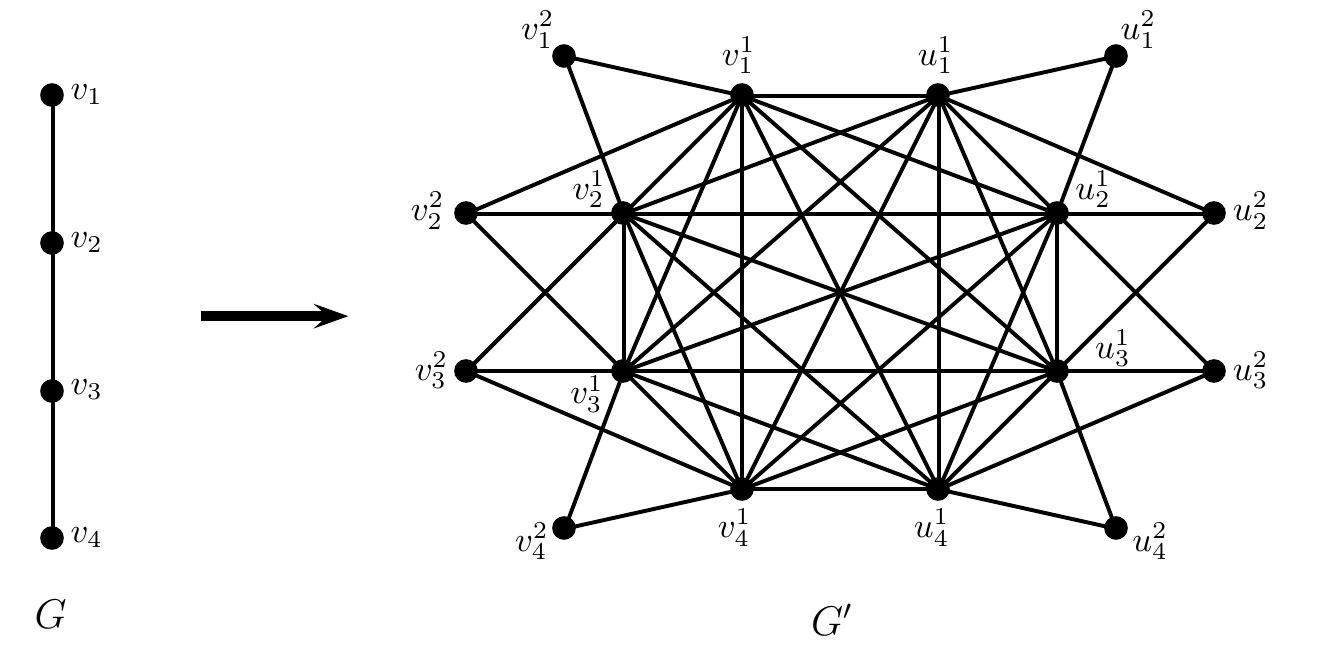}\\
 \caption{An illustration to the construction of $G'$ from $G$  in the proof of Theorem~\ref{t:split}.}
\label{fig:1}
\end{center}
\end{figure}

Now, to complete the proof of the theorem, we only need to prove the following claim.

 \begin{claim} $G$ has a dominating set of cardinality $k$ if and only if $G'$ has a semi-PD-set of size cardinality $2k$.
 \end{claim}
\begin{proof}
Let $D=\{v_{i_1},v_{i_2},\ldots,v_{i_k}\}$ be a dominating set of size atmost $k$ of $G$. Then $D_{sp}=\{v_{i_1}^{1},v_{i_2}^{1},\ldots,v_{i_k}^{1}\}\cup \{u_{i_1}^{1},u_{i_2}^{1},\ldots,u_{i_k}^{1}\}$ is a semi-PD-set of $G'$ of size atmost $2k$.

Conversely, suppose that $G$ has a semi-PD-set $D_{sp}$ of cardinality at most $2k$. Let $S_1=(V_1 \cup V_2) \cap D_{sp}$ and $S_2= U_1\cup U_2 \cap D_{sp}$. Then either $|S_1| \leq k$ or $|S_2| \leq k$. Without loss of generality, let us assume that $|S_1| \leq k$. Note that if $v_i^2 \in S_1$ and none of neighbors belong to $S_1$ then we replace $v_i^2$ by some of its neighbor $v_j^1$ in the set $S_1$. So, we may assume that $S_1 \cap V_2 = \phi$. Now the set $D= \{v_i \mid v_i^1 \in S_1\}$ is a dominating set of $G$ of size atmost $k$. Hence, the result follows.
\end{proof}
Hence, the theorem is proved.
\end{proof}


\section{Algorithm for Interval Graphs}
\label{Sec:5}

In this section, we present a linear-time algorithm to compute a minimum cardinality semi-PD-set of an interval graph.

A linear time recognition algorithm exists for interval graphs, and for an interval graph an interval family can also be constructed in linear time~\cite{booth,golumbic}. Let $G=(V,E)$ be an interval graph and $I$ be its interval model. For a vertex $v_{i}\in V$, let $I_{i}$ be the corresponding interval. Let $a_{i}$ and $b_{i}$ denote the left and right end points of the interval $I_{i}$.  Without loss of generality, we may assume that no two intervals share a common end point. Let $\alpha=(v_{1},v_{2},\ldots,v_{n})$ be the \emph{left end ordering} of vertices of $G$, that is,  $a_{i}<a_{j}$ whenever $i<j$. Now we first prove the following lemmas.

\begin{lemma}
Let $\alpha=(v_{1},v_{2},\ldots,v_{n})$ be the left end ordering of vertices of $G$. If $v_{i}v_{j}\in E$ for $i<j$, then $v_{i}v_{k}\in E$ for every $i<k<j$.
\end{lemma}
\begin{proof}
The proof directly follows from the left end ordering of vertices of $G$.
\end{proof}

Define the set $V_{i}=\{v_{1},v_{2},\ldots,v_{i}\}$, for each  $ i\in [n]$.

\begin{lemma}
If $G$ is a connected interval graph, then $G[V_{i}]$ is also connected.
\end{lemma}
\begin{proof}
The proof can easily be done using induction on $i$.
\end{proof}

Let $F(v_{i})$ be the least index vertex adjacent to $v_{i}$, that is, if $F(v_{i})=v_{p}$, then $p=\min\{k \mid v_{k}v_{i}\in E\}$. In particular, we define  $F(v_{1})=v_{1}$. Let $L(v_{i})=v_{q}$, where $q=\max \{k \mid v_{k}v_{i}\notin E$ and $k<i\}$. In particular, if $L(v_{i})$ does not exist, we assume that $L(v_{i})=v_{0}$ $(v_{0}\notin V)$. Let $G_{i}=G[V_{i}]$ and $D_{i}$ denote a semi-PD-set of $G_{i}$ of minimum cardinality. Recall that we only consider connected graphs with at least two vertices.

\begin{lemma}\label{lem3}
For $i\geq 2$, if $F(v_{i})=v_{1}$, then $D_{i}=\{v_{1},v_{i}\}$.
\end{lemma}
\begin{proof}
Note that every vertex in $G_{i}$ is dominated by $v_{1}$, and $d_{G_{i}}(v_{1},v_{i})=1$. Hence, $D_{i}=\{v_{1},v_{i}\}$.
\end{proof}

\begin{lemma}\label{lem4}
For $i>1$, if $F(v_{i})=v_{j}$, $j>1$ and $F(v_{j})=v_{1}$, then $D_{i}=\{v_{1},v_{j}\}$.
\end{lemma}
\begin{proof}
Note that every vertex in $G_{i}$ is dominated by some vertex in the set $\{v_{1},v_{j}\}$, and $d_{G_{i}}(v_{1},v_{i})=1$. Hence, $D_{i}=\{v_{1},v_{i}\}$.
\end{proof}

\begin{lemma}
\label{lem5}
For $r<k<j<i$, let $F(v_{i})=v_{j}$, $F(v_{j})=v_{k}$ $F(v_{k})=v_{r}$. If every vertex $v_{l}$ where $k<l<j$, is adjacent to at least one vertex in the set $\{v_{j},v_{r}\}$, then the following holds:\\
(a) $\{v_{j},v_{r}\}\subseteq D_{i}$.\\
(b) $v_{j}$ is semipaired with $v_{r}$ in $D_{i}$.\\
(c) $D_{i}\cap \{v_{s+1},\ldots,v_{r},v_{r+1},\ldots,v_{i}\}=\{v_{j},v_{r}\}$.
\end{lemma}
\begin{proof} (a) To dominate $v_{i}$, either $v_{i}\in D_{i}$ or $v_{i1}\in D_{i}$, where $j\leq i1 <i$ and $v_{i1}\in N_{G_{i}}(v_{i})$. If $i1\neq j$ and $v_{i1}$ is semipaired with some vertex $v_{j1}$, then $N_{G_{i}}(v_{i1})\subseteq N_{G_{i}}(v_{j})$, and $d_{G_{i}}(v_{j},v_{j1})\leq 2$. Hence, we can update the set $D_{i}$ as $D_{i}=(D_{i}\setminus \{v_{i1}\})\cup \{v_{j}\}$ and semipair $v_{j}$ with $v_{j1}$. This proves that $v_{j}\in D_{i}$.

If $v_{r}$ also belongs to $D_{i}$, then we are done. Otherwise, if $v_{j}$ is semipaired with $v_{j1}$ (where $j1\neq r$), then $j1>r$. Also, $N_{G}[v_{j1}]\subseteq N_{G}[v_{j}]\cup N_{G}[v_{r}]$. In that case, we can update the set $D_{i}$ as $D_{i}=(D_{i}\setminus \{v_{j1}\})\cup \{v_{r}\}$. Hence, $\{v_{j},v_{r}\}\subseteq D_{i}$.\\

\noindent(b) Suppose  $\{v_{j},v_{r}\}\subseteq D_{i}$. If $v_{j}$ is semipaired with $v_{r}$ in $D_{i}$, then we are done. Otherwise, if $v_{j}$ is not semipaired with $v_{r}$, assume that $v_{j}$ is semipaired with $v_{j1}$ and $v_{r}$ is semipaired with $v_{r1}$. Note that $j1$ must be greater than $r$, and $N_{G_{i}}[v_{j1}]\subseteq N_{G_{i}}[v_{j}]\cup N_{G_{i}}[v_{r}]$. Therefore, the set $D_{i}\setminus \{v_{j1}\}$ also dominates all the vertices of $G_{i}$.

Suppose that $N_{G_{i}}(v_{r1})\subseteq D_{i}$. In this case, $D'=D_{i}\setminus \{v_{j1},v_{r1}\}$ is a semi-PD-set of $G_{i}$ where $v_{j}$ and $v_{r}$ are semipaired. This contradicts the fact that $D_{i}$ is a semi-PD-set of $G_{i}$ of minimum cardinality. Hence, $N_{G_{i}}(v_{r1})\nsubseteq D_{i}$.

Let $v_{r2}\in D_{i}\setminus N_{G_{i}}(v_{r1})$. Now update the set $D_{i}$ as follows: remove $v_{j1}$ from $D_{i}$, add $v_{r2}$ in the set $D_{i}$, semipair $v_{j}$ with $v_{r}$ and $v_{r1}$ with $v_{r2}$. Clearly, the updated set is also a semi-PD-set of $G_{i}$ of minimum cardinality. This proves that there always exists a semi-PD-set $D_{i}$ of $G_{i}$ such that $\{v_{j},v_{r}\}\subseteq D_{i}$, and $v_{j}$ is semipaired with $v_{r}$ in $D_{i}$.\\

\noindent(c) We know that $\{v_{j},v_{r}\} \subseteq D_{i} \cap \{v_{s+1},\ldots,v_{r},v_{r+1},\ldots,v_{i}\}$. We need to show that $D_{i} \cap \{v_{s+1},\ldots,v_{r},\\v_{r+1},\ldots,v_{i}\} = \{v_{j},v_{r}\}$, that is, there is no other vertex from the set $\{v_{s+1},\ldots,v_{r},v_{r+1},\ldots,v_{i}\}$ belongs to $D_{i}$. Suppose, to the contrary, that there does not exist any $D_{i}$ for which $D_{i} \cap \{v_{s+1},\ldots,v_{r},\\v_{r+1},\ldots,v_{i}\}=\{v_{j},v_{r}\}$. So, for each $D_{i}$, $|D_{i} \cap \{v_{s+1},\ldots,v_{r},v_{r+1},\ldots,v_{i}\}|\geq 3$. Consider a set $D_{i}$ for which $|D_{i} \cap \{v_{s+1},\ldots,v_{r},v_{r+1},\ldots,v_{i}\}|$ is minimum.

Let $|D_{i} \cap \{v_{s+1},\ldots,v_{r},v_{r+1},\ldots,v_{i}\}|=l$. Also, assume that $v_{p}\in D_{i}$, where $p\neq j,r$ and $s+1\leq p\leq i$. Also, assume that $v_{p}$ is semipaired with $v_{p1}$ in $D_{i}$. Now consider the following two cases.

\noindent\emph{Case~1. $p1>s$.} If $v_{s}\in D_{i}$, then if, some vertex of the set $\{v_{1},v_{2},\ldots,v_{s}\}$ is dominated by $v_{p}$ or $v_{p1}$, then that vertex is also dominated by $v_{s}$. In that case, $D_{i}\setminus \{v_{p},v_{p1}\}$ is also a semi-PD-set of $G_{i}$, which is a contradiction. If $v_{s} \notin D_{i}$ and $N_{G_{i}}(v_{s})\subseteq D_{i}$, then also $D_{i}\setminus \{v_{p},v_{p1}\}$ is a semi-PD-set of $G_{i}$, which is again a contradiction. Hence, $v_{s} \notin D_{i}$ and $N_{G_{s}}(v_{s})\nsubseteq D_{i}$. Suppose $v_{q}\in N_{G_{s}}(v_{s})\cap D_{i}$. Then, update the set $D_{i}$ as $D_{i}=(D_{i}\setminus \{v_{p},v_{p1}\})\cup \{v_{s},v_{q}\}$. Note that $D_{i}$ is still a semi-PD-set of $G_{i}$ of minimum cardinality, and $|D_{i} \cap \{v_{s+1},\ldots,v_{r},v_{r+1},\ldots,v_{i}\}| < l$, a contradiction.

\noindent\emph{Case~2. $p1\leq s$.} If $v_{s}\notin D_{i}$, then the updated set $D_{i}=(D_{i}\setminus \{v_{p}\})\cup \{v_{s}\}$ is also a semi-PD-set of $G_{i}$ of minimum cardinality. If $v_{s}\in D_{i}$ and $N_{G_{s}}(v_{p1})\subseteq D_{i}$, then the updated set $D_{i}=D_{i}\setminus \{v_{p},v_{p1}\}$ is also a semi-PD-set of $G_{i}$, a contradiction. If $v_{s}\in D_{i}$ and $N_{G_{s}}(v_{p1})\nsubseteq D_{i}$, let $v_{q}\in N_{G_{s}}(v_{p1})\setminus D_{i}$. Then, update $D_{i}$ as $D_{i}=(D_{i}\setminus \{v_{p}\})\cup \{v_{q}\}$. Note that $D_{i}$ is still a semi-PD-set of $G_{i}$ of minimum cardinality, and $|D_{i} \cap \{v_{s+1},\ldots,v_{r},v_{r+1},\ldots,v_{i}\}|<l$, a contradiction.

Since both Case~1 and Case~2 produce a contradiction, there exists a semi-PD-set $D_{i}$ of $G_{i}$ of minimum cardinality, for which the set $D_{i}\cap \{v_{s+1},\ldots,v_{r},v_{r+1},\ldots,v_{i}\}$ contains only $v_{j}$ and $v_{r}$.
\end{proof}

\begin{lemma}
\label{lem6}
  For $r<k<j<i$, let $F(v_{i})=v_{j}$, $F(v_{j})=v_{k}$ $F(v_{k})=v_{r}$. If every vertex $v_{l}$ where $k<l<j$, is adjacent to at least one vertex in the set $\{v_{j},v_{r}\}$, then the following holds.\\
(a) $D_{i}=\{v_{j},v_{r}\}$ if $L(v_{r})=v_{0}$.\\
(b) $D_{i}=\{v_{1},v_{2},v_{j},v_{r}\}$ if $L(v_{r})=v_{1}$.\\
(c) $D_{i}=D_{s}\cup \{v_{j},v_{r}\}$ if $L(v_{r})=v_{s}$ with $s\geq 2$.
\end{lemma}
\begin{proof}
(a) Clearly $D_{i}=\{v_{j},v_{r}\}$.\\

(b) From Lemma~\ref{lem5}, we know that $\{v_{j},v_{r}\}\subseteq D_{i}$. Also, other than $v_{1}$, all vertices are dominated by the set $\{v_{j},v_{r}\}$. Hence, $D_{i}=\{v_{1},v_{2},v_{j},v_{r}\}$.\\

(c) Clearly $D_{s}\cup \{v_{j},v_{r}\}$ is a semi-PD-set of $G_{i}$. Hence $|D_{i}|\leq |D_{s}|+2$. We also know that there exists a semi-PD-set $D_{i}$ of $G_{i}$ of minimum cardinality such that $D_{i}\cap \{v_{s+1},v_{s+2},\ldots,v_{i}\} = \{v_{j},v_{r}\}$ (where $v_{j}$ and $v_{r}$ are semipaired in $D_{i}$). Hence $D_{i}\setminus \{v_{j},v_{r}\}\subseteq V(G_{s})$. Also, $\{v_{j},v_{r}\}$ dominates the set $\{v_{s+1},v_{s+2},\ldots,v_{n}\}$, implying that the set $\{v_{1},v_{2},\ldots,v_{s}\}$ is dominated by the vertices in $D_{i}\setminus \{v_{j},v_{r}\}$. Hence, the set $D_{i}\setminus \{v_{j},v_{r}\}$ is semi-PD-set of $G_{s}$. Therefore, $|D_{s}|\leq |D_{i}|-2$. This proves that $|D_{i}|=|D_{s}|+2$.
Hence, $D_{i}=D_{s}\cup \{v_{j},v_{r}\}$.
\end{proof}

\begin{lemma}
\label{lem7}
For $r<k<j<i$, let $F(v_{i})=v_{j}$, $F(v_{j})=v_{k}$ $F(v_{k})=v_{r}$, and  $\{v_{l}\mid k<l<j\} \nsubseteq N_{G_{i}}[v_r] \cup N_{G_{i}}[v_j]$. Let  $t=\max\{l \mid k<l<j$ and $v_{l}v_{j}\notin E \}$ (assume that such a $t$ exists). Let $F(v_{t})=v_{b}$. Then, the following holds.\\
(a) $\{v_{j},v_{b}\}\subseteq D_{i}$.\\
(b) $v_{j}$ is semipaired with $v_{b}$ in $D_{i}$.\\
(c) $D_{i}\cap \{v_{s+1},\ldots,v_{b},v_{b+1},\ldots,v_{i}\}=\{v_{j},v_{b}\}$.
\end{lemma}
\begin{proof} (a) First we show that $v_j \in D_i$. Suppose $v_j \notin D_i$. Let $v_p$ be the vertex dominating $v_i$ in $D_i$. Note that $j< p \leq i$ and $N_{G_{i}}[v_p] \subseteq N_{G_{i}}[v_j]$. Let $v_q$ be the vertex semipaired with $v_p$ in $D_i$. Since $N[v_p] \subseteq N[v_j]$, any vertex which is within distance~$2$ from $v_p$ is also within distance~$2$ from $v_j$. We can update $D_i$ as $D_i \setminus \{v_p\} \cup \{v_j\}$ with $v_j$ semipaired with $v_q$. Hence, $D_i$ contains $v_j$. Similarly, we can show that $D_i$ also contains $v_b$. So, $\{v_j, v_b\} \subseteq D_i$.\\

\noindent(b) If $v_{j}$ is semipaired with $v_{b}$ in $D_{i}$, then we are done. Suppose, to the contrary, that $v_{j}$ is not semipaired with $v_{b}$ in $D_{i}$. So, assume that $v_{j}$ is semipaired with $v_{p}$ and $v_{b}$ is semipaired with $v_{q}$ in $D_{i}$. We consider the four cases based on the values of the indices $p$ and $q$.

\noindent\emph{Case~1. $p>b$ and $q>b$.} Here, $N_{G_{i}}[v_{p}] \cup N_{G_{i}}[v_{q}] \subseteq N_{G_{i}}[v_{j}] \cup N_{G_{i}}[v_{b}]$. Hence, the set $D_{i}\setminus \{v_{p},v_{q}\}$ is also a semi-PD-set of $G_{i}$, a contradiction.

\noindent\emph{Case~2.  $p<b$ and $q<b$.} Since the distance between $v_{p}$ and $v_{j}$ is at most $2$, $p\geq r$. If $q<b$ and $d_{G_{i}}(v_{q},v_{b})\leq 2$, then $d_{G_{i}}(v_{q},v_{p})\leq 2$. So, in the set $D_{i}$, $v_{j}$ can be semipaired with $v_{b}$, and $v_{p}$ can be semipaired with $v_{q}$.

\noindent\emph{Case~3. $p>b$ and $q<b$.} Here, $N_{G_{i}}[v_{p}] \subseteq N_{G_{i}}[v_{j}] \cup N_{G_{i}}[v_{b}]$. If $N_{G_{i}}(v_{q})\subseteq D_{i}$, then the set $D_{i}\setminus \{v_{p},v_{q}\}$ is also a semi-PD-set of $G_{i}$, a contradiction. If $N_{G_{i}}(v_{q})\nsubseteq D_{i}$, let $v_{x}\in N_{G_{i}}(v_{q})\setminus D_{i}$. Then update $D_{i}$ as $D_{i}=(D_{i}\setminus \{v_{p}\})\cup \{v_{x}\}$, and semipair $v_{q}$ with $v_{x}$ and $v_{j}$ with $v_{b}$.

\noindent\emph{Case~4. $p<b$ and $q>b$.} Since  the distance between $v_{p}$ and $v_{j}$ is at most~$2$, $p\geq r$. Also $N_{G_{i}}[v_{q}] \subseteq N_{G_{i}}[v_{j}] \cup N_{G_{i}}[v_{b}]$. If $N_{G_{i}}(v_{p})\subseteq D_{i}$, then the set $D_{i}\setminus \{v_{p},v_{q}\}$ is also a semi-PD-set of $G_{i}$, a contradiction. If $N_{G_{i}}(v_{p})\nsubseteq D_{i}$, let $v_{y}\in N_{G_{i}}(v_{p})\setminus D_{i}$. Then update $D_{i}$ as $D_{i}=(D_{i}\setminus \{v_{q}\})\cup \{v_{y}\}$, and  semipair $v_{p}$ with $v_{y}$ and $v_{j}$ with $v_{b}$.

By the above four cases, there always exists a semi-PD-set $D_{i}$ of $G_{i}$ of minimum cardinality  such that $v_{j}$ is semipaired with $v_{b}$ in $D_{i}$. This completes the proof of part~(b).\\

\noindent(c) The proof is similar to the proof of Lemma~\ref{lem5}(c), and hence is omitted.
\end{proof}

\begin{lemma}\label{lem8}
For $r<k<j<i$, let $F(v_{i})=v_{j}$, $F(v_{j})=v_{k}$ $F(v_{k})=v_{r}$, and  $\{v_{l}\mid k<l<j\} \nsubseteq N_{G_{i}}[v_r] \cup N_{G_{i}}[v_j]$. Let  $t=\max\{l \mid k<l<j$ and $v_{l}v_{j}\notin E \}$ (assume that such a $t$ exists). Let $F(v_{t})=v_{b}$. Then, the following holds.\\
(a) $D_{i}=\{v_{j},v_{b}\}$ if $L(v_{b})=v_{0}$.\\
(b) $D_{i}=\{v_{1},v_{2},v_{j},v_{b}\}$ if $L(v_{b})=v_{1}$.\\
(c) $D_{i}=D_{s}\cup \{v_{j},v_{b}\}$ if $L(v_{b})=v_{s}$ with $s\geq 2$.
\end{lemma}

\begin{proof}
The proof is similar to the proof of Lemma~\ref{lem6}, and hence is omitted.
\end{proof}

Based on above lemmas, we present an algorithm to compute a minimum semi-PD-set of an interval graph.

\begin{algorithm}[H]
\textbf{Input:} An interval graph $G=(V,E)$ with a left end ordering $\alpha=(v_{1},v_{2},\ldots,v_{n})$ of vertices of $G$.\\
\textbf{Output:} A semi-PD-set $D$ of $G$ of minimum cardinality.\\
\small $V'=V;$\\
\While{$(V' \neq \phi)$}{
 Let $i=$ max$\{k \mid v_k \in V'\}$.
\If{$(F(v_i) = v_1)$}{
	 $D = D \cup \{v_1,v_i\};$\\
	 $V'= V' \setminus \{v_1,v_2, \ldots, v_i\};$}
\ElseIf{$(F(v_i) = v_j$ and $F(v_j) = v_1$ where $j>1)$}{
	 $D = D \cup \{v_1,v_j\};$
	 $V'= V' \setminus \{v_1,v_2, \ldots, v_i\};$}
\ElseIf{$(F(v_i) = v_j$ and $F(v_j) = v_k$ where $k \geq 2)$}{
	 Let $F(v_k) = v_r.$
	\If{$\{v_{k+1},v_{k+2},\ldots,v_{j-1}\}\subseteq N_{G}[v_{j}]\cup N_{G}[v_{r}]$}{
		\If {$(L(v_r) = v_0)$ }{
			 $D = D \cup \{v_j,v_r\};$\\
			 $V'= V' \setminus \{v_1,v_2, \ldots, v_i\};$}
		\ElseIf {$(L(v_r) = v_1)$ }{
			 $D = D \cup \{v_1,v_2,v_j,v_r\};$\\
			 $V'= V' \setminus \{v_1,v_2, \ldots, v_i\};$}
		\Else{
			 Let $(L(v_r) = v_s)$ where $s \geq 2$.\\
			 $D = D \cup \{v_j,v_r\}$;\\
			 $V'= V' \setminus \{v_{s+1},v_{s+2}, \ldots, v_i\};$
		}}
	\Else{
		 Let $t=$ max$\{l \mid k<l<j$ and $v_l \notin N_{G}(v_{j})\}$ and $F(v_t)=v_b$.
		\If {$(L(v_b) = v_0)$ }{
			 $D = D \cup \{v_j,v_b\};$\\
			 $V'= V' \setminus \{v_1,v_2, \ldots, v_i\};$}
		\ElseIf {$(L(v_b) = v_1)$ }{
			 $D = D \cup \{v_1,v_2,v_j,v_b\};$\\
			 $V'= V' \setminus \{v_1,v_2, \ldots, v_i\};$}
		\Else
        {
			 Let $(L(v_b) = v_s)$ where $s \geq 2$.\\
			 $D = D \cup \{v_j,v_b\};$\\
			 $V'= V' \setminus \{v_{s+1},v_{s+2}, \ldots, v_i\};$
		}
	}
}
}
\caption{\textbf{SEMI-PAIRED-DOM-IG(G)}}
\label{alg:seq}
\end{algorithm}

Here, we illustrate the algorithm SEMI-PAIRED-DOM-IG, with the help of an example. An interval graph $G$ and its interval model $I$ is shown in Fig~\ref{interval}.

\begin{figure}[h!]
 \begin{center}
  \includegraphics[width=12cm, height=5cm]{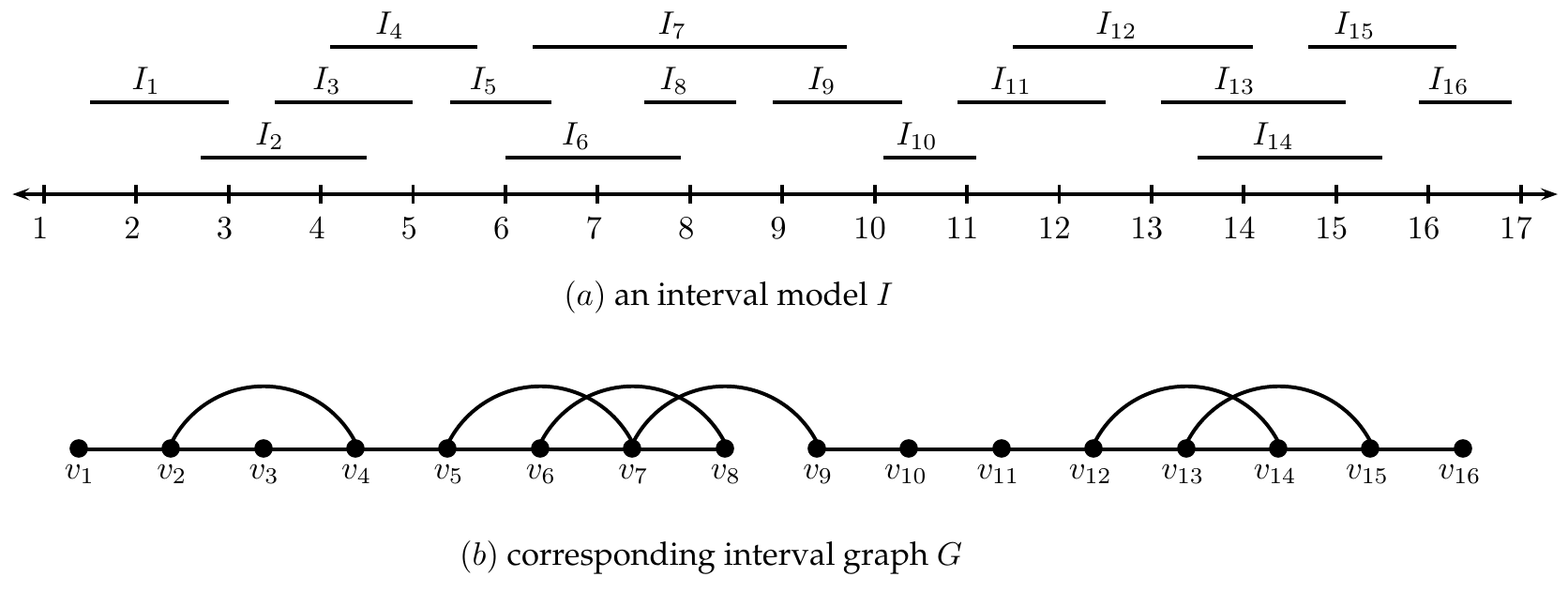}
 \caption{An interval model $I$ and corresponding interval graph $G$.}
\label{interval}
\end{center}
\end{figure}

For the interval graph $G$ given in Fig.~\ref{interval}, the algorithm SEMI-PAIRED-DOM-IG computes a semi-PD-set of minimum cardinality in $3$ iterations. Below, we illustrate all the $3$ iterations of the algorithm.\\


\[
\begin{tabular}{ |c| }
\hline
\small \underline{\textsc{Initially}}\\$V'=\{v_1,v_2, \ldots ,v_{16}\}$ and $D=\phi$.\vspace{0.15cm}\\
\hline
\large \underline{\textsc{Iteration $1$}}\\
$i=16$ and $F(v_i)=F(v_{16})=v_{15} \neq v_1$\\
$j=15$ and $F(v_j)=F(v_{15})=v_{13} \neq v_1$\\
$k=13$ and $F(v_k)=F(v_{13})=v_{12}$\\
$r=12$ and $\{v_{k+1} \ldots , v_{j-1}\}= \{v_{14}\} \subseteq N_{G}[v_j ] \cup N_{G}[v_r]$\\
Since $L(v_r)=L(v_{13})=v_{10}$ and $s=10 > 2$,\\
$D=D \cup \{v_{13},v_{15}\}$ and $V'=V' \setminus \{v_{11} \ldots v_{16}\}$.\\
\vspace{0.15cm}
\underline{\textsc{After Iteration $1$}}\\
$D=\{v_{13},v_{15}\}$ and  $V'=\{v_{1},v_{2} \ldots v_{10}\}$\vspace{0.15cm}\\
\hline
\vspace{0.15cm}
\large \underline{\textsc{Iteration $2$}}\\
$i=10$ and $F(v_i)=F(v_{10})=v_{9} \neq v_1$\\
$j=9$ and $F(v_j)=F(v_{9})=v_{7} \neq v_1$\\
$k=7$ and $F(v_k)=F(v_{7})=v_{5}$\\
$r=5$ and $\{v_{k+1} \ldots , v_{j-1}\} = \{v_{9}\}\nsubseteq N_{G}[v_j ] \cup N_{G}[v_r]$\\
In this case $t =$ max$\{l \mid k < l < j$ and $v_l \notin N_{G}(v_j)\}=8$ and \\$F(v_t) = F(v_8) = v_6$ $($clearly $b=6)$\\
Since $L(v_b)=L(v_{6})=v_{4}$ and $s=4 > 2$,\\
$D=D \cup \{v_{6},v_{9}\}$ and $V'=V' \setminus \{v_{5} \ldots v_{10}\}$.\\
\vspace{0.15cm}
\underline{\textsc{After Iteration $2$}}\\
$D=\{v_6,v_9, v_{13},v_{15}\}$ and  $V'=\{v_{1},v_{2},v_3,v_{4}\}$\vspace{0.15cm}\\
\hline
\vspace{0.15cm}
\large \underline{\textsc{Iteration $3$}}\\
$i=4$ and $F(v_i)=F(v_{4})=v_{2} \neq v_1$\\
$j=2$ and $F(v_j)=F(v_{2})=v_{1} $, hence\\
$D=D \cup \{v_{1},v_{2}\}$ and $V'=V' \setminus \{v_{1},v_{2},v_3,v_{4}\}$.\\
\vspace{0.15cm}
\underline{\textsc{After Iteration $3$}}\\
$D=\{v_1,v_2,v_6,v_9, v_{13},v_{15}\}$ and  $V'=\phi$\\
As $V'=\phi$ hence, loop terminates.\vspace{0.15cm}\\
\hline
\end{tabular}
\\ \vspace*{.3cm}
\]
Our algorithm returns the set $D=\{v_1,v_2,v_6,v_9, v_{13},v_{15}\}$, which is a minimum cardinality semi-PD-set of the interval graph $G$.
\begin{theorem}
Given a left end ordering of vertices of $G$, the algorithm SEMI-PAIRED-DOM-IG computes a semi-PD-set of $G$ of minimum cardinality in linear-time.
\end{theorem}

\begin{proof}
By Lemmas \ref{lem3}, \ref{lem4}, \ref{lem6} and \ref{lem8}, we can ensure that the algorithm SEMI-PAIRED-DOM-IG computes a semi-PD-set of $G$ of minimum cardinality. Also, it can be easily seen that the algorithm can be implemented in $O(m+n)$ time, where $n=|V(G)|$ and $m=|E(G)|$.
\end{proof}


\section{Algorithm for Trees}
\label{Sec:6}
In this section, we present a linear-time algorithm to compute a minimum cardinality semipaired dominating set in trees.

Let $T=(V,E)$ be a tree, and $\beta=(v_{n},v_{n-1},\ldots, v_{1})$ be the BFS ordering of vertices of $T$ starting at a pendant vertex $v_{n}$. Let $\alpha=(v_{1},v_{2},\ldots,v_{n})$ be the reverse ordering of $\beta$. In our algorithm, we process the vertices in the order they appear in $\alpha$. Let $p(v_{i})$ denote the parent of vertex $v_{i}$. If $v_{i}$ is the root vertex, we assume $p(v_{i})=v_{i}$.

The idea behind our algorithm is the following. We start with an empty set $D$, an array $L$ and an array $M$. Initially $L[v_{i}]=0$ and $M[v_{i}]=0$ for all $v_{i}\in V$. We process the vertices one by one in the order $\alpha=(v_{1},v_{2},\ldots,v_{n})$. During each of the iterations, we update $D$, $L$ and $M$ suitably. During the iterations, $L[v_{i}]=0$ if $v_{i}$ is not selected in $D$, $L[v_{i}]=1$ if $v_{i}$ is selected in $D$ but not semipaired, and $L[v_{i}]=2$ if $v_{i}$ is selected in $D$ and semipaired. Also, $M[v_{i}]=k$ if $v_{k}$ need to be semipaired with some vertex in $N_{T}[v_{i}]\setminus D$. At the end of the algorithm $D$ becomes a minimum cardinality semi-PD-set of the given tree $T$.
At the $i^{th}$ iteration, we process the vertex $v_{i}$. While processing $v_{i}$, we update $D$, $L$ and $M$ as follows.\\
\textbf{Case 1:}  $i\neq n,n-1$ and $v_{i}$ is not dominated by $D$.\\
\textbf{Subcase 1.1:} For every $v_{r}\in N_{T}[p(v_{i})]$, $M[v_{r}]=0$.\\
Update $D=D\cup \{p(v_{i})\}$, $L[p(v_{i})]=1$ and $M[p(v_{j})]=j$, where $v_{j}=p(v_{i})$.\\
\textbf{Subcase 1.2:} For some $v_{r}\in N_{T}[p(v_{i})]$, $M[v_{r}]\neq 0$. \\
Let $C=\{v_{r}\in N_{T}[p(v_{i})]\mid M[w]\neq 0 \}$. Let $v_{k}$ be the least index vertex in $C$ and $m[v_{k}]=v_{s}$. Update $L[p(v_{i})]=L[v_{s}]=2$, and $D=D\cup \{p(v_{i})\}$.\\
\textbf{Case 2:} $i \in \{n,n-1\}$ and  $v_{i}$ is not dominated by $D$.\\
Update $L[v_{n-1}]=L[v_{n}]=2$, and $D=D\cup \{v_{n-1},v_{n}\}$.\\
\textbf{Case 3:} $v_{i}$ is dominated by $D$ and $M[v_{i}]=0$.\\
No Update in $D$, $L$ and $M$ are made.\\
\textbf{Case 4:} $v_{i}$ is dominated by $D$ and $M[v_{i}]=k\neq 0$ (that is, $v_{k}$ need to be semipaired with some vertex in $N_{T}[v_{i}]\setminus D$).\\
\textbf{Subcase 4.1:} $L[p(v_{i})]=0$.\\
Update $L[p(v_{i})]=L[v_{k}]=2$, $M[v_{i}]=0$ and $D=D\cup \{p(v_{i})\}$.\\
\textbf{Subcase 4.1:} $L[p(v_{i})]=1$.\\
This case will not arrive.\\
\textbf{Subcase 4.3:} $L[p(v_{i})]=2$.\\
Update $L[v_{i}]=L[v_{k}]=2$, $M[v_{i}]=0$ and $D=D\cup \{v_{i}\}$.

\begin{theorem}
The \textsc{Minimum Semipaired Domination} problem is linear-time solvable in trees.
\end{theorem}

\section{Approximation Algorithm}\label{Sec:7}
In this section, we present a greedy approximation algorithm for the \textsc{Minimum Semipaired Domination} problem in graphs. We also provide an upper bound on the approximation ratio of this algorithm. The greedy algorithm is described as follows.

\medskip
\begin{algorithm}[H]
\caption{\textbf{:} APPROX-SEMI-PAIRED-DOM-SET(G)}
 \textbf{Input:} A graph $G=(V,E)$ with no isolated vertex.\\
\textbf{Output:} A semi-PD-set $D$ of $G$.\\
\small \Begin{
$D=\emptyset$;\\
$i=0$; $D_{0}=\emptyset$;\\
 \While{$(V\setminus (D_{0}\cup D_{1}\cup \ldots \cup D_{i}) \neq \emptyset)$}{
$i=i+1;$\\
choose two distinct vertices $u,v\in V$ such that $d_{G}(u,v)\leq 2$ and $|(N_{G}[u]\cup N_{G}[v])\setminus (D_{0}\cup D_{1}\cup\ldots \cup D_{i-1})|$ is maximized;\\
$D_{i}=(N_{G}[u]\cup N_{G}[v])\setminus (D_{0}\cup D_{1}\cup\ldots \cup D_{i-1})$;\\
$D=D\cup \{u,v\}$;\\
}
 return $D$;
}
\end{algorithm}

\begin{lemma}\label{lem:ap:1}
The algorithm APPROX-SEMI-PAIRED-DOM-SET produces a semi-PD-set of $G$ in polynomial time.
\end{lemma}
\begin{proof}
Clearly, the output set $D$ produced by the algorithm APPROX-SEMI-PAIRED-DOM-SET is a semi-PD-set of $G$. Also, each step of the algorithm can be computed in polynomial time. Hence, the lemma follows.
\end{proof}

\begin{lemma}\label{lem:ap:2}
For each vertex $v\in V$, there exists exactly one set $D_{i}$ which contains $v$.
\end{lemma}
\begin{proof}
We note that $V=D_{0} \cup D_{1}\cup \ldots D_{|D|/2}$. Also, if $v\in D_{i}$, then $v\notin D_{j}$ for $i<j$. Hence, the lemma follows.
\end{proof}

By Lemma~\ref{lem:ap:2}, there exists only one index $i\in [|D|/2]$ such that $v\in D_{i}$ for each $v\in V$. We now define $d_{v}=\frac{1}{|D_{i}|}$.
Now we are ready to prove the main theorem of this section.

\begin{theorem}
\label{t:thm4}
The \textsc{Minimum Semipaired Domination} problem for a graph $G$ with maximum degree $\Delta$ can be approximated with an approximation ratio of $1+\ln(2\Delta+2)$.
\end{theorem}
\begin{proof}For any finite set $X\neq \emptyset$, $\displaystyle\sum_{x\in X}\frac{1}{|X|}=1$. Hence, we have $$|D|=2\sum_{i=1}^{\frac{|D|}{2}}\sum_{w\in D_{i}}\frac{1}{|D_{i}|}=2\sum_{w\in V}d_{w}.$$

Let $D^{*}=\{u_{1},v_{1},u_{2},v_{2},\ldots,$ $u_{\frac{|D^{*}|}{2}},v_{\frac{|D^{*}|}{2}}\}$ be a semi-PD-set of $G$ of  minimum cardinality, where $u_{i}$ is semipaired with $v_{i}$, for each $i \in [\frac{|D^{*}|}{2}]$. Define $M = \{\{u_{1},v_{1}\},\{u_{2},v_{2}\},\ldots,\{u_{\frac{|D^{*}|}{2}},v_{\frac{|D^{*}|}{2}}\}\}$. Note that for each vertex $w$, there exists a pair $\{u_{i},v_{i}\}\in M$ such that $w\in N_{G}[u_{i}]\cup N_{G}[v_{i}]$. Hence, the following inequality follows.

$$\displaystyle\sum_{w\in V}d_{w}\leq \sum_{\{u_{i},v_{i}\}\in M} \sum_{w\in N_{G}[u_{i}]\cup N_{G}[v_{i}]} d_{w}.$$

Consider a pair $\{u,v\}\in M$ and define $z_{k}=|(N_{G}[u]\cup N_{G}[v])\setminus (D_{0}\cup D_{1} \cup D_{2} \cup \ldots D_{k})|$ for $k\in \{0\} \cup[\frac{|D|}{2}]$. Clearly, $z_{k-1}\geq z_{k}$ for $k\in [\frac{|D|}{2}]$. Suppose $l$ is the smallest index such that $z_{l}=0$. At the $k^{th}$ step of the algorithm, $D_{k}$ contains $z_{k-1}-z_{k}$ vertices from the set $N_{G}[u]\cup N_{G}[v]$. Hence

$$\displaystyle\sum_{w\in N_{G}[u]\cup N_{G}[v]}d_{w}=\sum_{k=1}^{l} (z_{k-1}-z_{k})\cdot \frac{1}{|D_{k}|}.$$

At the $k^{th}$ step of the algorithm, we choose the pair $u_{k},v_{k}$ such that $|D_{k}|=|(N_{G}[u_{k}]\cup N_{G}[v_{k}])\setminus (D_{0}\cup D_{1}\cup \cdots \cup D_{k-1})|$ is maximum. Hence $|D_{k}|\geq |(N_{G}[u]\cup N_{G}[v])\setminus (D_{0}\cup D_{1}\cup \cdots D_{k-1})|=z_{k-1}$. Therefore the following inequality follows.

$$\displaystyle\sum_{w\in N_{G}[u]\cup N_{G}[v]}d_{w}\leq \sum_{k=1}^{l} \frac{z_{k-1}-z_{k}}{z_{k-1}}.$$

For all integers $a<b$, we know that $H(b)-H(a)\geq \frac{b-a}{b}$, where $H(b)=\displaystyle\sum_{i=1}^{b}\frac{1}{i}$ and $H(0)=0$. Therefore
$$\displaystyle\sum_{w\in N_{G}[u]\cup N_{G}[v]}d_{w}\leq \sum_{k=1}^{l} H(z_{k-1})-H(z_{k})=H(z_{0})=H(|N_{G}[u]\cup N_{G}[v]|)\leq H(2\Delta+2).$$

It follows that
 $$|D|=2\displaystyle \sum_{w\in V}d_{w}\leq \sum_{\{u,v\}\in M}H(2\Delta +2)=|D^{*}|H(2\Delta +2)\leq (\ln(2\Delta+2)+1)\cdot |D^{*}|.$$

This shows that the \textsc{Minimum Semipaired Domination} problem can be approximated with an approximation ratio of $1+\ln(2\Delta+2)$.
\end{proof}

\section{Lower bound on approximation ratio} \label{Sec:8}
To obtain the lower bound on the approximation ratio  of the \textsc{Minimum Semipaired Domination} problem, we give an approximation preserving
reduction from the \textsc{Minimum Domination} problem. 
The following approximation hardness result is already known for the \textsc{Minimum Domination} problem.

\begin{theorem}\label{dom:ap}
\cite{chlebik}~For a graph $G=(V,E)$, the \textsc{Minimum Domination} problem cannot be approximated within $(1-\epsilon) \ln|V|$ for any $\epsilon > 0$ unless NP $\subseteq $ DTIME $(|V|^{O(\log\log|V|)})$.
\end{theorem}

Now, we are ready to prove the following theorem.

\begin{theorem}\label{approx-hard}
For a graph $G=(V,E)$, the \textsc{Minimum Semipaired Domination} problem cannot be approximated within $(1-\epsilon) \ln|V|$ for any $\epsilon > 0$ unless NP $\subseteq $ DTIME $(|V|^{O(\log\log|V|)})$.
\end{theorem}

\begin{proof}
Let $G=(V,E)$, where $V=\{v_{1},v_{2},\ldots,v_{n}\}$ be an arbitrary instance of the \textsc{Minimum Domination} problem. Now, we construct a graph $H=(V_{H},E_{H})$, an instance of the \textsc{Minimum Semipaired Domination} problem in the following way: $V_{H}=\{v_{i}^{1},v_{i}^{2},w_{i}^{1},w_{i}^{2},z_{i}\mid  i \in [n]\}$ and $E_{H}=\{w_{i}^{1}v_{j}^{1},w_{i}^{2}v_{j}^{2} \mid v_{j}\in N_{G}[v_{i}]\}\cup \{v_{i}^{1}v_{j}^{1},v_{i}^{2}v_{j}^{2}, z_{i}z_{j} \mid 1\leq i <j \leq n\}\cup \{v_{i}^{1}z_{j},v_{i}^{2}z_{j}\mid  i\in [n],j\in [n]\}$. Fig.~\ref{fig:3} illustrates the construction of $H$ from $G$.

\begin{figure}[h!]
 \begin{center}
  \includegraphics[width=11.5cm, height=4cm]{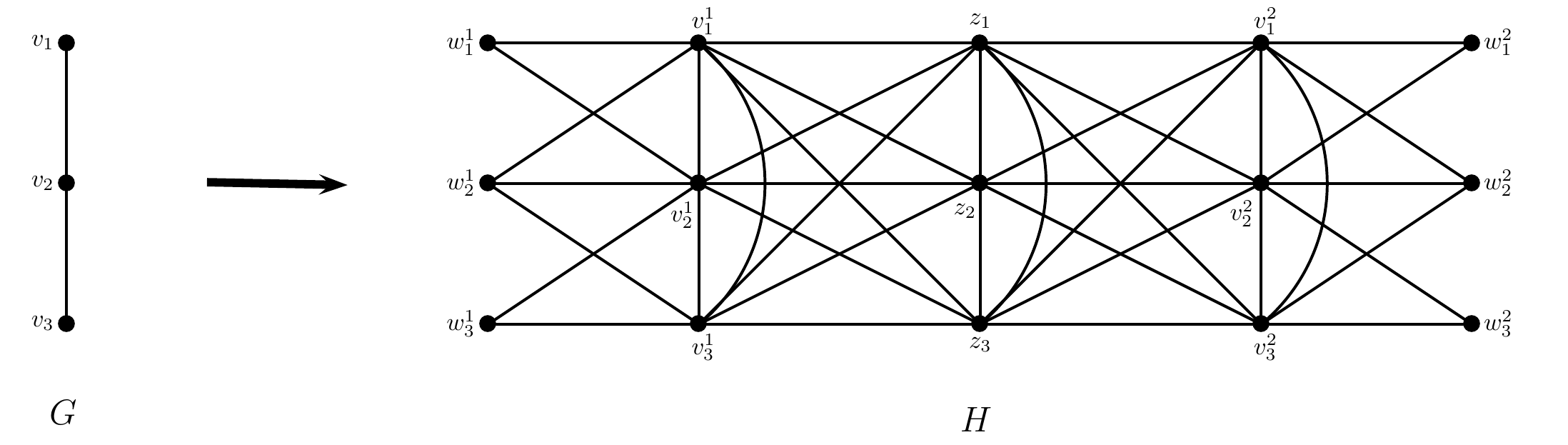}
 \caption{An illustration of the construction of $H$ from $G$  in the proof of Theorem~\ref{approx-hard}.}
\label{fig:3}
\end{center}
\end{figure}
Let $V^{k}=\{v_{i}^{k}\mid  i \in [n]\}$ and $W^{k}=\{w_{i}^{k}\mid i \in [n]\}$ for $k=1,2$. Also, assume that $Z=\{z_{i}\mid  i \in [n]\}$.
Note that $V^{1}\cup Z$ is a clique in $H$. Also $V^{2}\cup Z$ is a clique in $H$.

Let $D^{*}$ denote a minimum dominating set of $G$. Then the set $D'=\{v_{i}^{1},v_{i}^{2}\mid v_{i}\in D^{*}\}$ is a semi-PD-set of $H$. Hence, if $D_{sp}^{*}$ denotes a semi-PD-set of $H$ of minimum cardinality, then $|D_{sp}^{*}|\leq 2|D^{*}|$.

Suppose that the \textsc{Minimum Semipaired Domination} problem can be approximated within a ratio of $\alpha$, where $\alpha=(1-\epsilon)\ln(|V_{H}|)$ for some fixed $\epsilon>0$, by some polynomial time approximation algorithm, say \textbf{Algorithm A}. Next, we propose an algorithm, which we call \textbf{APPROX-DOMINATING-SET}, to compute a dominating set of a given graph $G$ in polynomial time.

\begin{algorithm}[H]
\caption{\textbf{: APPROX-DOMINATING-SET(G)}}
 \textbf{Input:} A graph $G=(V,E)$.\\
\textbf{Output:} A dominating set $D$ of $G$.\\
\small \Begin{
Initialize $k=0$;\\
Construct the graph $H$;\\
Compute a semi-PD-set $D_{sp}$ of $H$ using Algorithm~A;\\
Define $D_{sp}'=D_{sp}$;\\
\If {$(|D_{sp}'\cap (V^{1}\cup W^{1})|\leq |D_{sp}|/2)$}{
k=1;
}
\Else{
k=2;
}
\For {i=1 to n } {
\If {$(N_{H}(w_{i}^{k})\cap D_{sp}'==\emptyset)$} {
$D_{sp}'=(D_{sp}'\setminus w_{i}^{k})\cup \{v_{i}^{k}\}$;
}
}
$D=\{v_{i}\mid v_{i}^{k}\in D_{sp}'\cap V^{k}\};$\\
return $D$;
}
\end{algorithm}
Next, we show that the set $D$ returned by Algorithm $3$ is a dominating set of $G$. If $D_{sp}$ is any semi-PD-set of $H$, then clearly either $|D_{sp}\cap (V^{1}\cup W^{1})|\leq |D_{sp}|/2$ or $|D_{sp}\cap (V^{2}\cup W^{2})|\leq |D_{sp}|/2$. Assume that  $|D_{sp}\cap (V^{k}\cup W^{k})|\leq |D_{sp}|/2$ for some $k\in [2]$. Now, to dominate a vertex $w_{i}^{k}\in W^{k}$, either $w_{i}^{k}\in D_{sp}$ or $v_{j}^{k}\in D_{sp}$ where $v_{j}^{k}\in N_{H}(w_{i})$. If $N_{H}(w_{i}^{k})\cap D_{sp}$ is an empty set, then we update $D_{sp}$ by removing $w_{i}^{k}$ and adding $v_{j}^{k}$ for some  $v_{j}^{k}\in N_{H}(w_{i})$, and call the updated set $D_{sp}'$. We do this for each $i$ from $1$ to $n$. Note that even for the updated set $ D_{sp}'$, we have $|D_{sp}'\cap (V^{k}\cup W^{k})|\leq |D_{sp}|/2$. Also, in the updated set $D_{sp}'$, for each $w_{i}^{k}$, $N_{H}(w_{i}^{k})\cap (D_{sp}\cap V^{k})$ is non-empty. Hence $|D_{sp}'\cap V^{k}|\leq |D_{sp}|/2$ and $D_{sp}'\cap V^{k}$ dominates $W^{k}$. Therefore the set $D=\{v_{i}\mid v_{i}^{k}\in D_{sp}'\cap V^{k}\}$ is a dominating set of $G$. Also $|D|\leq |D_{sp}|/2$.

By above arguments, we may conclude that the Algorithm $3$ produces a dominating set $D$ of the given graph $G$ in polynomial time, and  $|D|\leq |D_{sp}|/2$. Hence, $ |D|\leq \frac{|D_{sp}|}{2}\leq \alpha \frac{|D_{sp}^{*}|}{2}\leq \alpha |D^{*}|.$

Also $\alpha=(1-\epsilon) \ln(|V_{H}|) \approx (1-\epsilon) \ln(|V|)$ where $|V_{H}|=5|V|$.  Therefore the Algorithm \textbf{APPROX-DOMINATING-SET} approximates the minimum dominating set within ratio $(1-\epsilon) \ln(|V|)$ for some $\epsilon>0$. By Theorem~\ref{dom:ap}, if the minimum dominating set can be approximated within ratio $(1-\epsilon) \ln(|V|)$ for some $\epsilon>0$, then NP $\subseteq $ DTIME $(|V|^{O(\log\log|V|)})$. Hence, if the \textsc{Minimum Semipaired Domination} problem can be approximated within ratio $(1-\epsilon) \ln(|V_{H}|)$ for some $\epsilon>0$, then NP $\subseteq $ DTIME $(|V_{H}|^{O(\log\log|V_{H}|)})$. This proves that the \textsc{Minimum Semipaired Domination} problem cannot  be approximated within $(1-\epsilon)\ln(|V_{H}|)$ unless NP $\subseteq$ DTIME $(|V_{H}|^{O(\log\log|V_{H}|)})$.
\end{proof}

\section{Conclusion}\label{Sec:9}
In this paper, we initiate the algorithmic study of the \textsc{Minimum Semipaired Domination} problem. We have resolved the complexity status of the problem for bipartite graphs, chordal graphs and interval graphs. We have proved that the \textsc{Semipaired Domination Decision} problem is NP-complete for bipartite graphs and split graphs. We also present a linear-time algorithm to compute a  semi-PD-set of  minimum cardinality for interval graphs and trees.  A $1+\ln(2\Delta+2)$ approximation algorithm for the \textsc{Minimum Semipaired Domination} problem in general graphs is given, and we prove that it can not be approximated within any sub-logarithmic factor. It will be interesting to study better approximation algorithms for this problem for bipartite graphs, chordal graphs and other important graph classes.

\end{document}